\pdfoutput=1
\documentclass[reprint, prx, twocolumn, superscriptaddress, floatfix, longbibliography]{revtex4-2}
\usepackage[utf8]{inputenc}
\usepackage{amsmath, amsthm, amssymb}
\usepackage{amsthm}
\usepackage{mathtools}
\usepackage{amsfonts}
\usepackage{tikz}
\usepackage{braket}
\usepackage[english]{babel}
\usepackage{comment}
\usepackage{subfigure}
\usepackage{graphicx} 
\usepackage[colorlinks=true,linkcolor=blue,urlcolor=blue,citecolor=blue]{hyperref}
\usepackage{algorithm}
\usepackage{algpseudocode}
\usepackage{ragged2e}
\usepackage{color}
\usepackage{array}
\usepackage{xcolor, colortbl}

\definecolor{green(munsell)}{rgb}{0.0, 0.7, 0.47}

\newcommand{\suchetana}[1]{{\leavevmode\color{green(munsell)}#1}}

\newcommand{\R}{\mathcal{R}}
\newcommand{\X}{\mathcal{X}}
\newcommand{\Y}{\mathcal{Y}}
\newcommand{\A}{\mathcal{A}}

\newcommand{\tr}{\text{Tr}}

\newcommand{\V}{\mathcal{V}}
\newcommand{\Pp}{\mathcal{P}}
\newcommand{\bellp}{\ket{\Phi^{+}}}
\newcommand{\mode}{\mathtt{mode}}
\newcommand{\Eval}{\mathtt{Eval}}
\newcommand{\Mes}{\mathtt{Measure}}
\newcommand{\Enc}{\mathtt{Enc}}
\newcommand{\Tr}{\text{Tr}}
\newcommand{\Prb}{\text{Pr}}

\newtheorem{theorem}{Theorem}
\newtheorem{lemma}{Lemma}

\newtheorem{definition}{Definition}

\newcounter{protocol}
\makeatletter
\newenvironment{protocol}[1][htb]{%
  \let\c@algorithm\c@protocol
  \renewcommand{\ALG@name}{Protocol}
  \begin{algorithm}[H]{%
  }}{\end{algorithm}
}

\begin{document}

\title{Hybrid Authentication Protocols for Advanced Quantum Networks}

\author{Suchetana Goswami}
\thanks{suchetana.goswami@gmail.com \\ These authors contributed equally to this work.}

\author{Mina Doosti}
\thanks{suchetana.goswami@gmail.com \\ These authors contributed equally to this work.}

\affiliation{School of Informatics, University of Edinburgh, 10 Crichton Street, Edinburgh EH8 9AB, United Kingdom}

\author{Elham Kashefi}
\affiliation{School of Informatics, University of Edinburgh, 10 Crichton Street, Edinburgh EH8 9AB, United Kingdom}
\affiliation{LIP6, CNRS, Sorbonne Universite, 4 Place Jussieu, Paris 75005, France}

\begin{abstract}
    Authentication is a fundamental building block of secure quantum networks, essential for quantum cryptographic protocols and often debated as a key limitation of quantum key distribution (QKD) in security standards. Most quantum-safe authentication schemes rely on small pre-shared keys or post-quantum computational assumptions. In this work, we introduce a new authentication approach that combines hardware assumptions, particularly Physical Unclonable Functions (PUFs), along with fundamental quantum properties of non-local states, such as local indistinguishability, to achieve provable security in entanglement-based protocols. We propose two protocols for different scenarios in entanglement-enabled quantum networks. The first protocol, referred to as the offline protocol, requires pre-distributed entangled states but no quantum communication during the authentication phase. It enables a server to authenticate clients at any time with only minimal classical communication. The second, an online protocol, requires quantum communication but only necessitates entangled state generation on the \textit{Prover}’s side. For this, we introduce a novel hardware module, the Hybrid Entangled PUF (HEPUF). Both protocols use weakly secure, off-the-shelf classical PUFs as their hardware module, yet we prove that quantum properties such as local indistinguishability enable exponential security for authentication, even in a single round. We provide full security analysis for both protocols and establish them as the first entanglement-based extension of hardware-based quantum authentication. These protocols are suitable for implementation across various platforms, particularly photonics-based ones, and offer a practical and flexible solution to the long-standing challenge of authentication in quantum communication networks.
\end{abstract}

\maketitle

\section{Introduction}
Recent advancements in quantum computation have been fundamentally shaped by the principles of secure quantum communication between spatially separated parties. In many cases, security arises from the inherent quantumness of the system, while classical correlations fail to provide the same level of protection~\cite{Terhal01, Eggeling02, bassoli_21, Lami21, broadbent_16, cacciapuoti19, caleffi18, dynes19, wehner18}. Perhaps the most well-known example is the information-theoretic security of quantum key distribution (QKD) protocols~\cite{bennett84, ekert91, bennett14, tomamichel17}, which rely solely on quantum mechanics as an accurate model of nature at certain scales. This relaxation of assumptions, basing security purely on the laws of physics, has opened new avenues for securing communication in a quantum world, one where protection against adversaries equipped with quantum computers and algorithms becomes crucial.\\

However, in practice, many quantum protocols, including QKD, impose additional requirements that necessitate extra effort and, in most cases, further assumptions. One of the most essential requirements in both classical and quantum communication is authentication: the ability to verify that a message or entity originates from a legitimate source. Authentication has frequently been criticised as one of the major limitations of QKD, with national agencies and governments highlighting it as a key reason for caution when considering quantum communication as the foundation for next-generation cryptographic systems~\cite{nsaquantum,ncsc}. Nonetheless, as argued in~\cite{renner2023debate}, authentication is a universal challenge that extends beyond quantum protocols, always requiring either a pre-shared secret, a trusted third party, or a computational assumption, whether in classical or quantum networks. Thus, practical and reliable authentication mechanisms are vital for the future of secure communication as a whole. Since authentication inherently relies on some form of assumption \cite{Ling25}, this opens the door to designing authentication schemes based on alternative assumptions, such as hardware or physical assumptions~\cite{arapinis21,doosti21,chakraborty23}. An important consideration when integrating these schemes with quantum communication is ensuring that the protocols retain the strong security guarantees and provability offered by quantum cryptographic techniques as much as possible. This requirement often makes the security proof particularly challenging.\\

In hardware-based authentication protocols, the security is partially provided by exploiting the underlying physical devices and their hardware assumptions. A Physically Unclonable Function (PUF) \cite{gassend02_1, gassend02_2, lee04} is a physical device that satisfies certain desirable security properties such as unpredictability or unforgeability, by hardware assumptions, due to its intrinsic physical randomness established during the manufacturing process~\cite{ruhrmair11}. It is also assumed that the creation of an identical clone of a given PUF is fundamentally infeasible, as even the slightest variation in the manufacturing process results in a distinctly different PUF. Due to these properties, PUFs offer a range of applicability in this domain, especially as a unique device fingerprint~\cite{ruhrmair11, guajardo07, kim18}. A PUF can be interacted with physically, formalised often in a query setting, and produce different outcomes. The set of these inputs and outputs is referred to as \textit{challenge-response pairs} (CRPs). While the literature about the CPUF is quite well-explored \cite{gassend02_1, guajardo07, kim18, roel12}, it is shown to be vulnerable against machine learning modelling-based attacks \cite{becker15_1, becker15_2, delvaux19, ruhrmair10, ruhrmair13}.\\

These vulnerabilities have motivated the study of this concept in the quantum world, where in~\cite{arapinis21} it has been shown unlike the classical setting, provably secure Quantum Physical Unclonable Function (QPUF) can exist. Several protocols have also been designed based on QPUF or a combination of CPUF together with quantum communication~\cite{kumar21, doosti21, galetsky22, chakraborty23, ghosh2024existential, davidson2024airqkd, smith2023fast, khan2023soteria}. In particular,~\cite{chakraborty23} introduced an authentication protocol that leverages quantum communication to achieve a provable exponential security advantage, even when using an underlying weak classical PUF. This hybrid design is advantageous as it enhances both the practicality of the protocol and its compatibility with existing quantum communication infrastructure. However, the proposed protocol follows a prepare-and-send network model, leaving open the intriguing question of \emph{whether an entanglement-based variant of such an authentication protocol can exist and what advantages it might offer}. This question is compelling both from a foundational perspective and a practical one. From the foundational point of view, by deepening our understanding of the quantum resources required for secure communication, and from a practical standpoint, particularly in advanced quantum networks where nodes are enhanced with distributed entanglement. Although prepare-and-send quantum network protocols are often less resource-intensive than their entanglement-based counterparts, they are not always the preferred choice for implementation across all platforms. Notably, many demonstrations of QKD protocols in photonic settings rely on entangled states, which offer enhanced resilience against environmental noise~\cite{Chapman20, Chapman22, yu25, zhang24, cai17, Kogias17, Zhou18, Keet10, tagliavacche24}.\\

In this work, we propose the first hybrid hardware-based authentication protocols that use entangled states as their main quantum resource. The term hybrid~\footnote{Hybrid protocols are sometimes used to refer to schemes that combine post-quantum assumptions with quantum communication, or even in other contexts, which is not the intended meaning here.} in this context refers to the combination of classical hardware components and assumptions with quantum communication. More precisely, our protocols implement the identification functionality, where the objective is to verify that communication is occurring with an honest, authenticated party. However, we use the term authentication interchangeably, as in our network setting, a network link can be authenticated in this manner. One of our main motivations for designing authentication protocols based on entangled states is the distinguishability properties of their subsystems.\\

Quantum state discrimination plays a crucial role in various cryptographic protocols~\cite{Terhal01, Eggeling02, Lami21, Bandyopadhyay21, Markham08, Rahaman15}. It involves identifying an unknown quantum state from a known set. If a system, distributed among multiple spatially separated parties, is prepared from a set of pure orthogonal states, a global measurement, i.e., a joint entangling measurement in the right basis, can perfectly identify it. However, when parties are limited to local operations and classical communication (LOCC), the same states may become indistinguishable~\cite{Bennett00, Ghosh01, Walgate02, Ghosh02, Horodecki03, Fan04, Ghosh04, Watrous05, Hayashi06, Bandyopadhyay11, Yu12, Halder18, Halder19, Goswami23}. Such a set of states is called Locally Indistinguishable (LI) and has applications in secret sharing and data hiding~\cite{Terhal01, Eggeling02, Lami21, Markham08, Rahaman15, Goswami23}. A set of two pure orthogonal states is always distinguishable via LOCC~\cite{Walgate00}, provided classical communication is allowed. However, in the cryptographic world we are interested in, spatially separated parties are ideally confined to their local labs, naturally restricting global measurements. Our protocols exploit the properties of local indistinguishability as a key ingredient, ensuring security from the limits of local state discrimination. We show that using these properties simplifies both completeness and security proofs, making them more compact in some cases.

We propose two authentication protocols leveraging entangled states shared between \textit{Verifier} ($\V$) and a \textit{Prover} ($\Pp$). As is common in PUF-based protocols, the \textit{Verifier} often has access to a database of challenge-response pairs (CRPs) derived from a PUF, while the \textit{Prover} possesses the physical device needed to prove their authenticity. During the authentication phase, the \textit{Verifier} sends a classical challenge to which the \textit{Prover} must respond correctly, following the protocol’s rules. The \textit{Verifier} then verifies the expected response. The security of such protocols intuitively ensures that no adversary (classical or quantum) can pass these verification tests without access to the actual device, even if they fully control the communication channel. This property is formally captured by the cryptographic notion of unforgeability~\cite{doosti2021unified,alagic2020quantum,doosti2022unclonability}. In our case, the adversary is capable of arbitrary quantum attacks, meaning the protocol must maintain security even against powerful quantum network adversaries.\\

Our first protocol extends PUF-based hybrid authentication into an entanglement-based setting, in the most straightforward way. Here, a trusted source pre-distributes entangled states before the protocol begins. Once authentication starts, it relies solely on classical communication, referred to thus as \emph{offline protocol}. This design is well-suited for quantum network architectures where nodes share entangled states, allowing authentication to occur at any time without additional quantum communication. The protocol has potential applications, such as serving as a ping test in advanced quantum networks without consuming extra resources. We first prove security for the ideal version of the protocol and then extend our analysis to scenarios where an adversary can tamper with the distributed entanglement.\\

The second protocol is conceptually different and arguably more intriguing from a quantum information perspective, as it exploits local indistinguishability more effectively as a core security feature. Unlike the first protocol, it does not require pre-distributed entanglement from a trusted party. Instead, entangled states are generated dynamically during the protocol using our proposed Hybrid Entangled PUF (HEPUF) construction, which encodes responses into bipartite entangled states. A subsystem is then transmitted over the channel as part of the authentication process, leading us to call this \emph{the online protocol}. Through a detailed security analysis, we prove that this protocol is secure against adaptive quantum adversaries who are unbounded over the channel and polynomially bounded in the setup phase.\\

Our protocols provide both provable security and practical authentication solutions for current and future quantum networks, offering implementation advantages such as enhanced robustness across specific platforms and encoding schemes. Beyond their direct applications, the novel combination of quantum resources, particularly local indistinguishability, opens new research directions in quantum communication and protocol design, creating new possibilities for designing more advanced resource-efficient quantum cryptographic schemes.

\section{Preliminaries} 
\noindent We first review the formal definitions of a classical PUF which we use as the underlying component of our protocols and constructions, as well as a short description of the quantum analogue of them known as quantum PUF (QPUF) as defined in~\cite{arapinis21}, and more recent hybrid constructions defined in~\cite{chakraborty23}. We also review the quantum information properties that we use in this work.\\

\noindent \emph{Models for PUF.} A Physical Unclonable Function (PUF) is a secure hardware cryptographic device that is, by assumption, hard to clone or reproduce, even for the manufacturer. Mathematically, classical PUFs are usually defined as a black-box with some underlying probabilistic functions, due to their inherent physical randomness. Such functions produce consistent random bits (with some distributions) for a given input, and in that regard are different from a random number generator. We recall the definition from~\cite{chakraborty23} for such functions and their randomness property:

\begin{definition}[Probabilistic Function (from~\cite{chakraborty23})]
A \emph{probabilistic function} is a mapping $f: \R \times \X \rightarrow \Y$ with an input space $\X$, a \emph{random coin} space $\R$, and an output space $\Y$. 
\end{definition}
\noindent For a fixed input $x \in \X$, and a random coin (or key) $R \leftarrow \R$, we define the probability distribution of the output random variable $f(x) := f(R,x)$ over all $ y \in \Y$ as,

\begin{equation}
    p^f_x(y) := \Prb[f(x) = y|x] = \sum_{r:f(r,x) = y} \Prb[R = r].
\end{equation}

A classical PUF can be modelled as a probabilistic function $f:\R \times\X\rightarrow\Y$ where $\X$ is the input space, $\Y$ is the output space of $f$ and $\R$ is the identifier. The creation of a classical PUF is formally expressed by invoking a manufacturing process $f\leftarrow\mathcal{MP}_{C}(\lambda)$, where $\lambda$ is the security parameter. An important property of a classical PUF in this model is the notion of $\emph{randomness}$, which is the maximal probability of $p^f_x(y)$ with an input $x_j\in \X$ on PUF $f_i$ where $i\in\R$. 

\begin{definition}[$p$-Randomness (from~\cite{chakraborty23})]\label{def:p-randomness}
We define the $p$-randomness of a classical PUF $f:\R \times \X\rightarrow\Y$ as \begin{equation}
    p := \max_{\substack{x \in \X \\ y \in \Y}} p^f_x(y).
\end{equation}
\end{definition}

Apart from these, a device with an underlying function $f$ needs to satisfy certain requirements to be qualified as a PUF. There exist different sets of such requirements in the literature~\cite{armknecht16, arapinis21, doosti21, chakraborty23}, however, for this work we adopt the minimal set of requirements also used in hybrid schemes such as in~\cite{chakraborty23}, which includes: $\delta_{1}$-Robustness, $\delta_{2}$-Collision Resistance, and $\delta_{3}$-Uniqueness. We omit the formal definition of these requirements for the purpose of this work, as we are not explicitly using them here, and whenever referring to classical PUF, we assume these requirements are satisfied. \\

\noindent \emph{Quantum and Hybrid PUFs.} Quantum PUF was first introduced as a cryptographic tool in \cite{arapinis21}. A quantum PUF can be considered as a quantum process that responds with a quantum state when challenged with an input quantum state. Quantum PUFs exploit quantum properties to enhance security and resistance against cloning or simulation, however, they come with several shortcomings in practice, such as the complexity of implementation and often the requirement of large quantum memories. To achieve a more practically accessible scheme, in \cite{chakraborty23}, the authors develop a hybrid authentication protocol based on a combination of a classical PUF and quantum communication. In this case, in every challenge-response pair, the challenge is considered to be classical, and the hybrid device responds with a sequence of single-qubit quantum states in every round. The security relies on the uncertainty properties of BB84 states as well as a technique used to boost a stronger adaptive model known as \emph{quantum lock}. The strength of the protocol lies in the simplicity of the construction, as well as the strong security guarantees. However, the security proofs are complicated due to the combination of the classical and quantum-enhanced forgery performed by the adversary.\\

\noindent \emph{Local indistinguishability.} The security feature of our online protocol relies on a very basic, but fundamental task in quantum information theory, i.e. state discrimination. Note that, the states in the set $S_1$ (Eq.(\ref{set})) are exactly identifiable by performing a two-qubit measurement in the Bell basis on the whole system. But in our case, the concerned parties are restricted to perform LOCC. It has already been shown that the set $S_1$ is LI \cite{Ghosh01, Halder18}, as there exists no non-trivial positive operator valued measure (POVM) that can distinguish between the states coming from the particular set, when applied on local subsystems. As the states are LI, this intrinsically ensures that no information about the state can be revealed by having access to only one qubit. Now, in our protocol, we consider only a pair of states out of four, either $\{ \ket{\Phi^+},\ket{\Psi^-} \}$ or $\{ \ket{\Phi^-},\ket{\Psi^+} \}$. Here, the choice of the pair is crucial. It is known that any two pure orthogonal states are always distinguishable via LOCC \cite{Walgate00}. However, for this state discrimination process, the parties must communicate the basis of the measurement along with the outcomes. In the following lemma, we show how to make two pure states LI and hence useful for our protocol.

\begin{lemma}
    When the states $\ket{\Phi^+}$ $(=1/\sqrt{2}(\ket{00}+\ket{11}))$, and $\ket{\Psi^-}$ $(=1/\sqrt{2}(\ket{01}-\ket{10}))$ are distributed among two spatially separated parties, who are allowed to perform measurements either in the computational or in the Hadamard basis probabilistically (governed by the underlying classical primitive), then they are locally indistinguishable when the exact basis of measurement is not revealed. 
    \label{lem_LI}
\end{lemma}

\begin{proof}
    Let us consider that there are two spatially separated parties $A_1$ and $A_2$, holding one qubit each and the joint state is prepared in either of the two states $\{\ket{\Phi^+}, \ket{\Psi^-}\}$. The goal here is to identify the state via LOCC without revealing the basis of their corresponding measurement in a given round. To identify, let $A_2$ performs a measurement in the computational basis $\{ \ket{0},\ket{1}\}$ (which is determined by some underlying classical primitive) and obtains outcome $b_{A_2}=0$ (with local state collapsed to $\ket{0}\bra{0}$) or $b_{A_2}=1$ (with local state collapsed to $\ket{1}\bra{1}$) with equal probability due to the structure of the states. Similarly, they can also perform the measurement in the Hadamard basis $\{ \ket{+},\ket{-}\}$. In this case too, the outcomes are $b_{A_2}=0$ when the local state collapses to $\ket{+}\bra{+}$) or $b_{A_2}=1$ when the local state collapses to $\ket{-}\bra{-}$). Even when the outcomes are announced, as the basis is not revealed, for $A_1$, the local state remains a maximally mixed state; hence, exact identification is not possible. 
\end{proof}
\noindent It is evident that similar results can be obtained for states $\{ \ket{\Phi^-},\ket{\Psi^+} \}$.\\

\noindent \emph{Adversarial model and security definitions.} The most relevant adversarial model for communication protocols such as PUF-based authentication is the so-called \emph{network adversary} model. In this adversarial model, it is often assumed that the main parties (which, for authentication, we refer to as the \emph{Verifier} and \emph{Prover}) are honest, but the communication channel is insecure — i.e., there is an adversary on the channel who can arbitrarily access and manipulate the communications. For some protocols in the literature, this is also referred to as the man-in-the-middle attack~\cite{fei2018quantum}. If such adversaries have quantum capabilities, they can be described by either an arbitrary quantum operation on the channel (unbounded quantum adversaries) or as Quantum Polynomial Time (QPT) adversaries, which are any arbitrary quantum algorithms as long as they are polynomially bounded. In an authentication or identification protocol, the goal of the adversary is to fake themselves as the \textit{Prover} — i.e., to forge the identity or the message sent by the \textit{Prover}. For PUF-based protocols, however, there are other considerations that need to be taken into account. First, PUFs are often theoretically assumed to be secure by definition, even against the manufacturer of the device. However, in this work, as we work with weak and realistic PUF models, we cannot assume this. If the manufacturer colludes with the network adversary, the protocols are trivially broken due to the weakness of the classical PUF. Hence, we assume that the PUF or our following constructions based on PUFs are manufactured according to given instructions, and that the underlying classical PUFs are not \emph{fully} broken by any means — including corruption by the manufacturer.

An authentication or identification protocol typically consists of three phases: the \emph{set-up} phase, the \emph{identification} phase, and the \emph{verification} phase. In the setup phase, all the parameters of the protocol are set, and all the prior information required for the main phase is gathered. This means that in a PUF-based protocol, during the setup phase, the CPUF/HPUF/QPUF is queried by the verifier, allowing them to construct a database of CRPs that will later be used for authentication. Then, the physical device is sent to the \textit{Prover}. In general, we assume that this transition happens over a public and potentially insecure channel. Hence, it is often assumed that the adversary also has their own local database of CRPs of the PUF, which they can later use to attack the protocol. Since this interaction with the device is assumed to be allowed for the adversary, no PUF, classical or quantum, can achieve security against unbounded quantum/classical adversaries, as any such devices are eventually characterizable by an exponential number of queries~\cite{arapinis21}. Thus, the most powerful possible adversary to consider is the class of QPT adversaries. Nonetheless, we show that quantum communication allows for an advantage here, which is that, apart from the assumption on the setup phase, the protocol can achieve security against unbounded channel adversaries during the authentication phase. The authentication phase is when the interaction between the \textit{Verifier} and \textit{Prover} happens, leading to some inputs from the \textit{Prover}, which will be used in the verification phase by the verifier to certify that the party or message is authentic. We note that the possible class of attacks performed by the adversary here is quite large, and includes faking communication from the \textit{Verifier} and attempting to extend their database by pretending to be the \textit{Verifier}. As such, this includes any \emph{adaptive strategy} for the adversary. Hence, proving security in such strong models is challenging, and sometimes some schemes can only achieve security against a weaker version of network adversaries known as \emph{non-adaptive} or static adversaries~\cite{chakraborty23}, where the adversary cannot choose the queries and instead receives $q$ challenge-response pairs.

The main security property of an authentication protocol is \emph{unforgeability}. Informally, unforgeability means that the adversary is unable to produce a new valid input-output pair (CRP for PUFs) that can pass verification after seeing a subset of the inputs and outputs. The unforgeability for classical PUFs has been defined in~\cite{armknecht16}. Unforgeability has also been extended to the quantum setting~\cite{doosti2021unified,boneh2013quantum,arapinis21,chakraborty23}, both for PUFs and other schemes. Here, we use the notion of unforgeability known as \emph{universal unforgeability}, defined in~\cite{doosti2021unified,arapinis21}, which is an average-case definition over the choice of the inputs. We skip the formal definition here and refer to~\cite{doosti2021unified,chakraborty23} for the formal game and definition.
\begin{definition}[Universal Unforgeability (informal)]
\label{def:sec1}
A PUF/HPUF, with verification algorithm $\mathtt{Ver}(.)$ provides universal unforgeability against an adaptive QPT adversary $\A$ if the probability of the adversary producing a response $r$ (classical or quantum, depending on the scheme) for a randomly chosen challenge $c$, successfully passing $\mathtt{Ver}$ is negligible with respect to the security parameter. 
\end{definition}
This intuitively means that the adversary cannot predict the behaviour of the PUF after seeing $q$ queries that are chosen adaptively, as long as the number of queries is polynomially bounded. This is the most practically relevant notion of unforgeability for authentication protocols.

Finally, going from the security of the primitive to the protocol, we characterise the \emph{Correctness} (or \emph{Completeness}) and the \emph{Soundness} (or security) of protocols. \emph{Correctness} means that the success probability of any honest party successfully passing the verification is close to $1$ with negligible error. \emph{Soundness} of any verification protocol ensures that the success probability of any adversary (depending on the adversarial model) to pass the overall process of verification is bounded by $\epsilon$, which is often expected to be a negligible function of the security parameter. The completeness and soundness can be defined per round or over $M$ rounds of the protocol in total.

\section{Our new hybrid construction}
\noindent \emph{Hybrid Entangled PUF (HEPUF)}. We now introduce a new hybrid PUF, which takes a classical input and produces an entangled quantum output. Note that, we only use this construction for our \emph{online protocol}. 

In the spirit of hybrid PUF in~\cite{chakraborty23}, we also start with a CPUF with a p-randomness value as defined in Def. \ref{def:p-randomness}. An HEPUF produces a bipartite entangled pure state as an output, which encodes the outcome of the CPUF. Our general definition of HEPUF, in fact, covers a family of constructions, characterised by the selected set of entangled states. One obvious choice, for this set for instance, is the family of Bell states: 

\begin{equation}
S_1: 
\begin{cases}
\ket{\Phi^+}=\frac{1}{\sqrt{2}}(\ket{00} + \ket{11})= \frac{1}{\sqrt{2}}(\ket{++} + \ket{--})\\
\ket{\Psi^+}=\frac{1}{\sqrt{2}}(\ket{01} + \ket{10})= \frac{1}{\sqrt{2}}(\ket{++} - \ket{--})\\
\ket{\Phi^-}=\frac{1}{\sqrt{2}}(\ket{00} - \ket{11})= \frac{1}{\sqrt{2}}(\ket{+-} + \ket{-+}) \\
\ket{\Psi^-}=\frac{1}{\sqrt{2}}(\ket{01} - \ket{10})= \frac{1}{\sqrt{2}}(\ket{-+} - \ket{+-})
\end{cases}
\label{set}
\end{equation}

We can now define HEPUFs as follows:

\begin{definition}[HEPUF]\label{hepuf}
    Let $f:\X \rightarrow \Y$ be a classical PUF with an $n$-bit input and $m$-bit output i.e. $f:\{0,1\}^n \rightarrow \{0,1\}^{m}$. For every input $x$, we have $y=f(x) = y^1||y^2$, where the concatenation assigns the first $l^1$ bit to $y^1$ and the rest of the $l^2 = m - l^1$ bit to $y^2$~\footnote{In most cases $y^1$ and $y^2$ have $[\frac{m}{2}]$ bits each}. A HEPUF then is defined as a 5-tuple of $(\mode, f, S, M, \Eval(.), \Mes(.))$, where $\mode$ is a variable that shows the \emph{mode of operation} of the HEPUF, $S$ is the set of two-qubit entangled states used for the encoding, $M$ is a set of measurement used for local decoding, $f$ is a CPUF as described above. Let $k = \lfloor l^1 / |S| \rfloor$ be the size of the quantum output of HEPUF. Then $\Eval$ and $\Mes$ are algorithms that run as follows, on classical input $x$:

    \begin{equation}
    HEPUF :=
        \begin{cases}
          \mode = 0 ~~~~ \begin{split}      
          &\Mes(.)=\perp \\
          &\Eval(x) = f(x)\\          
          & \Eval(x) \leftarrow HEPUF
          \end{split}
          \\\\
          \mode = 1 ~~~~ \begin{split}      
          &\Mes(.)=\perp \\
          &\Eval(x) = \Enc_S(y^2),\\
          & \Enc_S(y_i^2) = \ket{\Psi^{y^2_i}_S}_{VP}, \\
          & \rho^i_V = \Tr_P[\ket{\Psi^{y^2_i}_S}\bra{\Psi^{y^2_i}_S}_{VP}]\\
          &[\otimes^k_{i=1} \rho^i_V] \leftarrow HEPUF
          \end{split}
          \\\\
          \mode = 2 ~~~~ \begin{split}      
          &\Eval(.) = \perp\\
          &\Mes(1) := b = M_{y^1}(\rho_P),\\
          & b_i \in \{0,1\}, b = b_i \dots b_k ,\\
          & \rho^i_P = \Tr_V[\ket{\Psi^{y^2_i}_S}\bra{\Psi^{y^2_i}_S}_{VP}],\\
          & \rho_P = [\otimes^k_{i=1} \rho^i_P]\\
          & b \leftarrow HEPUF
          \end{split}
        \end{cases}
    \end{equation}    
    
    \noindent where $\Enc_S(y^2)$ encodes the bits of $y^2$ into a series of bipartite states $\ket{\Psi^{y^2_i}_S}_{VP}$, or in other words, assigns an element of the set $S$ to the output value of the CPUF's first subsection. $M_{y^1}$ refers to the measurement of $P$ subsystem according to $y^1$, resulting to bit-string b.
\end{definition}

The above definition gives a general construction for HEPUFs, which can be instantiated with different sets of entangled states and POVMs to capture a broad family of such hardware modules. Informally, every HEPUF works as follows: The HEPUF has three general modes that are controlled by the value $\mode$. The $\mode = 0$ is a one-time mode that is only used in the setup phase by the \textit{Verifier}, and the HEPUF can never go back to this mode after the variable is changed per round. This is similar to one of the mechanisms proposed in~\cite{chakraborty23} for the HPUF, however, here we explicitly embed this mode in the definition. Then in $\mode = 1$, the HEPUF evaluates the underlying CPUF, generates the two parts of the output $y^1$ and $y^2$ and, given the pre-defined set of bipartite states $S$, encodes the second half of the response tuple i.e. $y^2$ into these states. Here $k$ of such states are generated for every input $x$. Then the HEPUF, keeps the \textit{Prover}'s subsystem of these states and outputs the $V$ subsystem as the quantum outcome. Hence the quantum outcome of the HEPUF in this mode, consists of the tensor product of $k$ mixed states. In $\mode = 2$, the HEPUF measures the subsystem of $\Pp$ using a measurement set $M$, with respect to the first part of the response tuple i.e. $y^1$. This means that $y^1$ usually determines the measurement basis. 

In the specific instantiation of this construction that we use for our online protocol, we use the set $S = \{\ket{\Phi^+}, \ket{\Psi^-}\}$ (or equivalently $\{\ket{\Psi^+}, \ket{\Phi^-}\}$) denoting two of the four states in the Bell basis. We fix the measurement set to be the single qubit measurement in the conjugate basis (computational and Hadamard basis). As such, the above construction becomes rather simple and implementable with many of the existing quantum communication technologies.


\section{Offline Authentication Protocol with Bell States}

\begin{figure}[htp]
    \fbox{\includegraphics[scale=0.14]{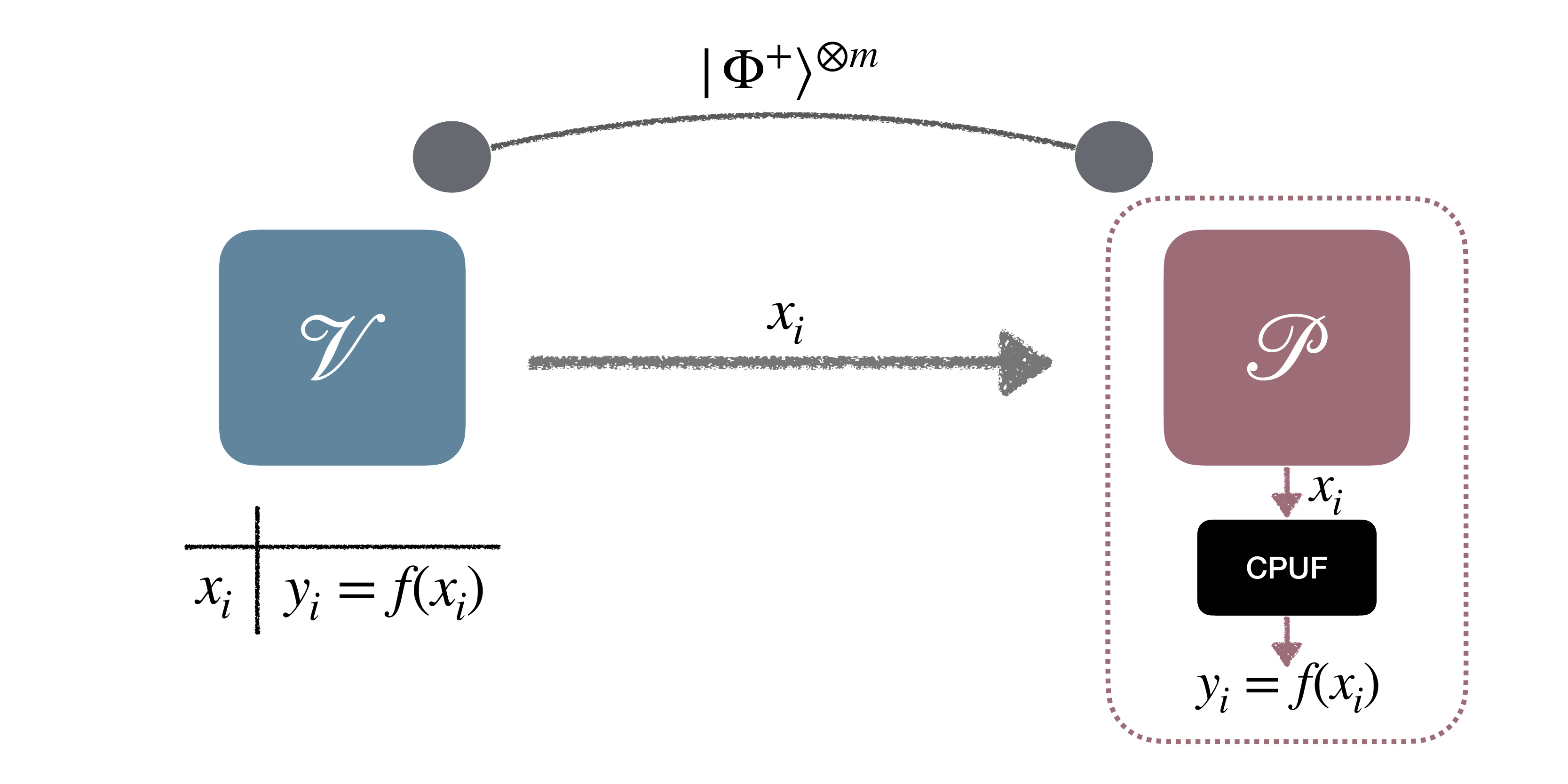}}
    \caption{\footnotesize{Schematic for the offline protocol with pre-shared Bell states from a trusted source.}}
    \label{off_schematic}
\end{figure}

The first class of authentication protocol that we propose using CPUF, is a simple protocol that uses pre-shared and trusted Bell states. Here, no quantum communication is required between the parties during the process of authentication. We refer to this protocol as an \emph{Offline Authentication Protocol} since in a quantum network with enabled entangled nodes, the parties can use this protocol to authenticate each other at any point in time, without requiring any additional quantum communication, and only using LOCC. This protocol can be seen as the first entanglement-based analogue of the hybrid CPUF-based authentication protocol proposed in \cite{chakraborty23}. 

Informally, the protocol proceeds as follows (also depicted in Fig.~\ref{off_schematic}): The \textit{Verifier} ($\V$) and the \textit{Prover} ($\Pp$) use a CPUF with $m$-bit length output. Similar to the usual setup of the PUF-based protocols, we assume that the \textit{Verifier} has access to an efficient-size classical database of the inputs and outputs of the PUF, which can be acquired during the setup phase, and the \textit{Prover} has access to the device itself. Here, the parties also receive $m$ copies of $\bellp$ from a trusted source. The choice of the $\ket{\Phi^+}$ state for this protocol is arbitrary, and the protocol works with any one of the Bell states in set $S_1$ in Eq. (\ref{set}). Note that, the protocol consumes $m$ copies of $\bellp$ for every round of authentication with input $x$ and an $m$-bit output of CPUF. During the authentication phase, for the $i$-th bit, the \textit{Verifier} selects a challenge $x$ at random from the database and sends it to the \textit{Prover}. The authentic \textit{Prover} can obtain the outcome $f(x)$ by querying $x$ from the CPUF. Then $\Pp$ measures their share of $\bellp$ per bit, in a particular basis chosen based on $[f(x_i)]$ and records the corresponding outcome as a bit-string $a = a^1\dots a^i \dots a^m$. Finally, $\V$ can deterministically authenticate $\Pp$ by locally measuring their share of the Bell-pairs and comparing her outcome $b$  to the string $a$. The formal description of the protocol is given in Protocol~\ref{prot_bell}.

Before providing the security proof of this protocol, we point out a few important properties of this protocol. First, we note that the protocol only uses classical communication between the two parties, and no quantum communication is required. The entangled states are distributed at the beginning of the protocol, and the parties only perform LOCC. This is why this protocol is suited to offline quantum networks with pre-shared entanglement amongst the nodes, where at any point in time that the authentication is required, the parties can consume as many numbers of states as required for the authentication. In an experimental domain, this property plays a crucial role, as it makes the process of authentication fast. Another interesting property of the protocol, especially when realised with a PUF, is that, unlike almost all the PUF-based protocols, the response of the PUF, i.e. $y_i$, is never sent over the communication channel. This implies that, in the vanilla version of the protocol where the entangled source is trusted and the states are in fact Bell-pairs, the adversary cannot gather learning data to break the PUF from eavesdropping on the communication channel. Although this argument is not yet enough for the security proof of the protocol, it provides a good intuition of why such entangled-based version PUF protocols are interesting to explore.

\begin{protocol}
    \begin{minipage}{0.45\textwidth}
    \justifying
    \noindent The protocol runs between two parties, a \textit{Verifier} ($\V$) and a \textit{Prover} ($\Pp$), and includes a trusted distributor to share entangled states between the two parties prior to the protocol. The protocol includes a CPUF $f:\X \rightarrow \Y$ with an $n$-bit input and $m$-bit output, as the hardware resources.

    \begin{enumerate}
        \item \textbf{Setup phase:}
        \begin{enumerate}
            \item $\V$ and $\Pp$ receive their share of the Bell states, they consume $\bellp^{\otimes m}$ copies per every authentication round per challenge.
            \item $\V$ constructs a database of pairs $\mathtt{DB}:= \{(x_i,y_i)\}_{i=1}^{d}$ from PUF where $y_i = f(x_i)$. 
        \end{enumerate}
        \item \textbf{Authentication phase:}
        \begin{enumerate}
            \item \label{a}$\V$ chooses a pair $(x_i,y_i) \in \mathtt{DB}$ uniformly at random (per round), and sends $x_i$ to $\Pp$ over the public classical channel.
            \item \label{b}$\Pp$ receives $x_i$ and obtains $y_i = f(x_i)$ by interacting with the CPUF locally. 
            \item \label{c} $\texttt{For j=0 to j=m}$:\\
            $\texttt{If}$ the jth bit of $y_i$, i.e $y^j = [f(x_i)]^j = 0$, \texttt{then:} $\Pp$ measures their qubit of one of the $\bellp$ in the computational basis ($\{\ket{0},\ket{1}\}$), and records the measurement outcome in the bit $a^j$.\\
            \texttt{Otherwise:} if $y^j = 1$, $\Pp$ measures their qubit of the Bell state in the Hadamard basis ($\{\ket{+},\ket{-}\}$), and records the measurement outcome in the bit $a^j$.
            \item \label{d} After all the entangled states have been measured, $\Pp$ sends the string $a_i = a^1\dots a^j \dots a^m$ to $\V$ over the public classical channel.
        \end{enumerate}
        \item \textbf{Verification phase:}
        \begin{enumerate}
            \item \label{e} $\V$ also measures their share of the Bell states per each bit of the outcome in the basis determined by $y_i$, similar to $\Pp$, and records the outcome bits in the string $b_i = b^1\dots b^j \dots b^m$.\\
            \texttt{If:} $a_i = b_i$, the $\V$ accepts, \texttt{otherwise}, $\V$ rejects and aborts the protocol.
        \end{enumerate}
    \end{enumerate}
    \end{minipage}
    \caption{Entanglement-based offline authentication}
    \label{prot_bell}
\end{protocol}

Finally, an interesting comparison can be drawn between this protocol and the HLPUF protocol in \cite{chakraborty23}. The HLPUF-based authentication protocol is a BB84-style quantum protocol which encodes the PUF's output into quantum states to enhance the security of insecure classical PUFs. Here, our protocol resembles the E91 protocol and does not rely on encoding the PUF outcomes, as they are being used for local measurement instead. Note that, unlike the HPUF-based prepare and measure protocol \cite{chakraborty23}, the authentication process of this protocol relies on the correct choice of the basis for measurement. This choice is solely governed by the response obtained by querying the authentic database. Also, the classical communication from the honest \textit{Prover} to the \textit{Verifier} only reveals the outcome of the measurement, which is random and independent of the choice of the measurement basis.\\

We now provide the completeness and soundness of the protocol in the following. 
\begin{theorem}
    Protocol~\ref{prot_bell} satisfies completeness and exponential security against any unbounded quantum eavesdropper, given trusted perfect Bell-pairs. 
\end{theorem}
\begin{proof}
\emph{Completeness:} According to the definition of completeness, we need to show that the success probability of an honest \textit{Prover} in the absence of an adversary over the channel is overwhelmingly close to one. Since $\V$, and $\Pp$ share a perfect $\bellp$ state, and since they both measure in the same basis (determined by the outcome bit of the PUF, $y_i$), their measurement outcomes are perfectly correlated, i.e $\forall j: a_j = b_j$. As the verification of $\V$ only relies on this equality testing, which is given as,
    \begin{equation}
        \Prb[\V \text{ accept}_{\Pp}] = \Prb[\texttt{acc} \leftarrow \V(a^j) | \text{no Adv}: \forall j] = 1
        \label{eq:p1-completness}
    \end{equation}

\emph{Security:} We show that the success probability of an unbounded quantum eavesdropper, and any overall QPT adversary, performing general forgery attacks is bounded. Given the iid assumption on the PUF's outcome bits, we bound the probability for one bit and then straightforwardly extend it to an $m$-bit outcome PUF. Since each Bell state is perfect and shared prior to the authentication phase by a trusted party, the adversary has no access to the entangled state. As a result, the adversary cannot perform any quantum strategies over the channel. The only possible attack viable for the adversary ($\A$) is intercepting the classical communication and sending their own bit $\Tilde{a}^j$ (instead of $a^j$) to pass the verification process. According to definition of security (in Def.~\ref{def:sec1}) we have,
\begin{equation}
\begin{split}
    \Prb[\V \text{accept}_{\A}] & = \Prb[\texttt{acc} \leftarrow \V(\Tilde{a^j}) |  \Tilde{a^j} \leftarrow \A: \forall j]\\
    & = \Prb[b^j = \Tilde{a}^j : \forall j]\\
    & = \Prb[\A \text{ guessing } a^j : \forall j]
\end{split}
    \label{eq:p1-sec-1}
\end{equation}
For each bit $j$, the probability of $\A$ correctly guessing $a^j$ is exactly $1/2$, even for a PUF with bias value $\delta$. Since here the adversary has no access to the encoded outcome of the PUF into the Bell state, and even if they could for the classical PUF's bit, with probability $p^c_{forge}$, they cannot manipulate the state and hence the measurement outcome of the \textit{Verifier} remains purely random, and $\Prb[\A \text{ guessing } a^j] = 1/2$. Hence over $m$-bits we get,
\begin{equation}
\begin{split}
    \Prb[\V \text{ accept}_{\A}] = \Prb[b^j = \Tilde{a}^j : \forall j] = \left(\frac{1}{2}\right)^m = negl(\lambda).
\end{split}
    \label{eq:p1-sec-1}
\end{equation}
This concludes the proof.
\end{proof}

Even though this vanilla version of the protocol achieves exponential security, (that too) with one round of authentication for an $m$-bit PUF, the trust assumptions on the pre-shared Bell-states are rather strong. Especially assuming the perfect $\bellp$, and that the adversary has no access to this quantum resource, is a strong, and in many cases, unrealistic assumption. One might ask whether the protocol still achieves reasonable security in the presence of an imperfect Bell state, where the adversary can also entangle their subsystem with the state shared between the \textit{Verifier} and the \textit{Prover}. In the spirit of security proofs of entanglement-based QKD protocols using entropic uncertainty principles, we show that the protocol still achieves security. Since the completeness still remains the same, we only give the security proof in the following theorem. 


\begin{theorem}
    Protocol~\ref{prot_bell} satisfies soundness against any unbounded quantum adversary, even when the shared entangled pairs are imperfect. More precisely, suppose that for each round of authentication, if the joint state $\rho_{VP}$ shared between the Verifier and the honest Prover satisfies $F(\rho_{VP},\ket{\Phi^+}\bra{\Phi^+}) \ge 1-\epsilon$, then the probability that any adversary successfully forges all $m$ bits of a protocol execution is upper bounded by,
    \[
    \Prb[\V \text{ accept}_{\A}] \;\le\; 2^{-m(1-\mu(\epsilon))},
    \]
    and it is negligible in the security parameter $\lambda$.
\end{theorem}

\begin{proof}
We analyse one round of authentication in details and then extend the analysis for $m$ rounds by entropy accumulation. Let, $\rho_{VPA}$ be the joint state of the \textit{Verifier} $\V$, the honest \textit{Prover} $\Pp$, and the adversary $\A$ before the measurement.
By assumption,
\begin{equation}
F(\rho_{VP},\ket{\Phi^+}\bra{\Phi^+}) \ge 1-\epsilon ,
\end{equation}
with $\ket{\Phi^+}=(\ket{00}+\ket{11})/\sqrt{2}$ and $\rho_{VP}=\tr_\A [\rho_{VPA}]$.\\

\noindent Let $Y\in\{0,1\}$ be the output bit of the CPUF, based on which the \textit{Verifier} chooses the measurement basis (either computational $Z_V$ or Hadamard $X_V$).
Let $B^V$ be the \textit{Verifier}’s measurement outcome and let $A$ denote all quantum side information held by the adversary. 
By the entropic uncertainty relation with quantum side information~\cite{berta2010uncertainty}, we have 
\begin{equation}
    H_{\min}(X_V|A,Y) + H_{\max}(Z_V|P,Y) \ge 1 .
\end{equation}
where the right hand side is $log_2 (1/c)$ with $c=1/2$. Since $F(\rho_{VP},\ket{\Phi^+}\bra{\Phi^+})\ge 1-\epsilon$, by the continuity bound for conditional max-entropy~\cite{winter2016tight}, for two states $\rho_{AB}$ and $\sigma_{AB}$ with fidelity at least $1-\epsilon$, this bound states that:
\begin{equation}
\bigl| H_{\max}(A|B)_\rho - H_{\max}(A|B)_\sigma \bigr|
\le 2\sqrt{\epsilon}\log d_A + h_2(2\sqrt{\epsilon}),
\end{equation}
where $d_A$ is the dimension of system $A$ and
$h_2(p)$ is the binary entropy function. In our setting $A=\V$, is a qubit, hence $d_A=2$, and $H_{\max}(Z_V|P)_{\Phi^+}=0$ for the ideal Bell state.
Therefore,
\begin{equation}
H_{\max}(Z_V|P,Y) \le \mu(\epsilon),
\quad
\mu(\epsilon)=2\sqrt{\epsilon}+h_2(2\sqrt{\epsilon}).
\end{equation}
Consequently,
\begin{equation}
H_{\min}(B^V|A,Y) \ge 1-\mu(\epsilon).
\end{equation}
For $\epsilon\le 2^{-\lambda}$, $\mu(\epsilon)$ is negligible in the security
parameter $\lambda$.

\noindent Now we find the bound on the adversary’s guessing probability. The optimal probability that the adversary correctly predicts the \textit{Verifier}’s bit in one round is given by, 
\begin{equation}
\Pr[\tilde{B}^V = B^V]
  \le 2^{-H_{\min}(B^V \mid A,Y)}
  \le \tfrac{1}{2} + \mu(\epsilon),
\end{equation}

\noindent We extend our analysis for $m$ rounds, min-entropy accumulation theorem.
Let $B_1^V,\ldots,B_m^V$ be the \textit{Verifier}’s outcomes in $m$ independent rounds and $Y_1,\ldots,Y_m$ the corresponding CPUF bits. We have:
\begin{equation}
\begin{split}
    H_{\min}(B_1^V\cdots B_m^V \mid A,Y_1\cdots Y_m) & \ge \sum_{i=1}^m H_{\min}(B_i^V \mid A,Y_i) \\
    & \ge m(1-\delta(\epsilon)).
\end{split}
\end{equation}
\noindent Therefore the probability that the adversary forges all $m$ bits is bounded by,
\begin{equation}
\Pr_{forge}
  \le 2^{-H_{\min}(B_1^V\cdots B_m^V \mid A,Y_1\cdots Y_m)}
  \le 2^{-m(1-\mu(\epsilon))}.
\end{equation}
\noindent For a polynomial $m$, this is a negligible function in $\lambda$, which concludes the proof.  
\end{proof}

\section{Online Authentication Protocols with Higher-dimensional States}

\noindent The second class of protocols that we discuss here also uses bipartite maximally entangled states, but this time, in an online way. The sharing of entanglement only takes place during the process of authentication. As quantum communication is involved in the protocol, we call this protocol an online one. This essentially indicates that the protocol does not require any pre-shared entanglement between the parties. However, the main motivation for this protocol is to exploit a unique information-theoretic property of quantum states, known as \emph{local indistinguishability}. The online protocol relies on HEPUF introduced in Def. \ref{hepuf}, which essentially involves the set of states introduced in Eq. \ref{set}. \\

Here, we provide a sketch of the protocol flow, as also depicted in Fig.\ref{on_schematic}. This protocol uses an HEPUF construction, as defined in Definition\ref{hepuf}, instead of relying directly on a CPUF. The \textit{Verifier} ($\V$) and the \textit{Prover} ($\Pp$) each have access to an HEPUF built from a weak CPUF with $n$-bit input and $2^m$-bit output. During authentication, in the $i$-th round, the \textit{Verifier} sends a challenge $x_i$ to the \textit{Prover}. The honest \textit{Prover} inputs this challenge into the HEPUF. The second part of the response, denoted $y_i^1 || y_i^2$ (specifically, the $y_i^2$ portion), is internally encoded into a bipartite entangled state, and a corresponding qubit is output as the response. The \textit{Prover} then sends this subsystem to the \textit{Verifier} via the quantum channel. Next, the \textit{Prover} instructs the HEPUF to perform a projective measurement on their local subsystem, with the measurement basis determined by the value of $y_i^1$. The outcomes of these measurements for each individual encoded bit are recorded and publicly announced over a classical channel. This measurement causes the \textit{Verifier}’s state to collapse into a specific state. Finally, the \textit{Verifier} measures their local qubit and performs a verification procedure using both their own measurement outcome and the \textit{Prover}’s announcement, in order to accept or reject the \textit{Prover}.

\begin{figure}[htp]
    \fbox{\includegraphics[scale=0.14]{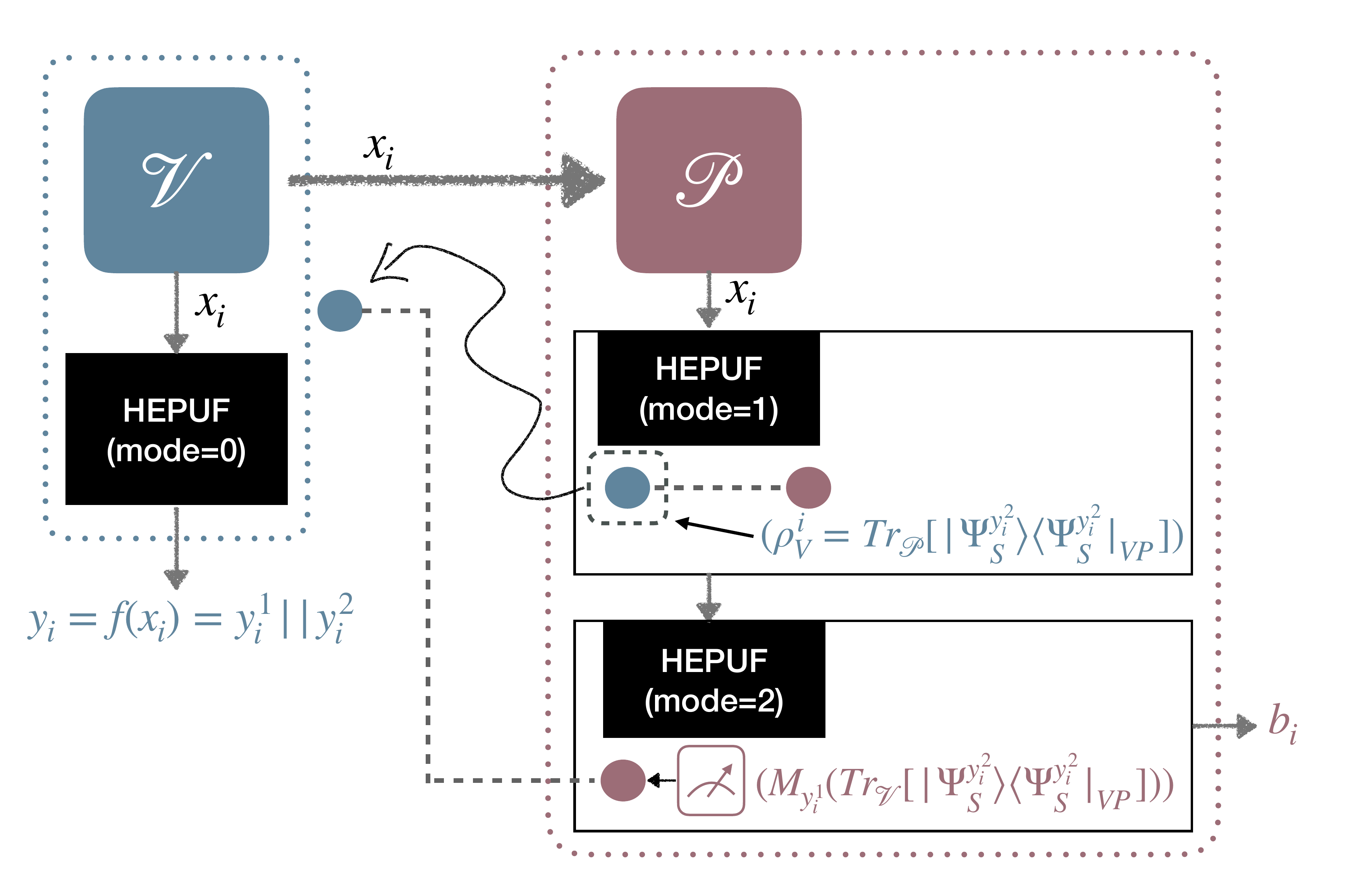}}
    \caption{\footnotesize{Schematic for the online protocol with HEPUF.}}
    \label{on_schematic}
\end{figure}

Note that, unlike the offline case, this protocol does not involve any trusted source of pre-shared entanglement. Here, the security of the protocol relies on the properties of quantum state discrimination as the HEPUF encodes the response corresponding to each challenge to a quantum state according to the Def.~\ref{hepuf}. Here, the security proof of this authentication protocol is developed based on the observation made in Lemma~\ref{lem_LI}. 

\noindent \emph{Specification of the Construction.} This protocol uses a specific instantiation of the HEPUF as defined in Def.~\ref{hepuf}. First, the outcome $y = f(x) = y^1||y^2$ is concatenated equally. Every bit of $y^2$, as mentioned before, the HEPUF outputs either of the states, and a projective measurement is performed depending on the value of $y^1$. It can be easily seen, a similar protocol can be designed with the other two Bell states $\{ \ket{\Phi^-}, \ket{\Psi^+} \}$. In the following part of the paper, we discuss the completeness and security of the online protocol involving $\{ \ket{\Phi^+}, \ket{\Psi^-}\}$. The same results follow for the other two states. 

\begin{theorem}
    Protocol \ref{online_prot_2Bell} satisfies completeness.
\end{theorem}

\begin{proof}
    For completeness, similar to the offline protocol, we show that the success probability of an honest \textit{Prover} in the absence of an adversary over the channel is overwhelmingly close to one. The protocol is such that, when a part of the response tuple, i.e. $y^2_i=0$, the measurement outcomes of $\V$ and $\Pp$ are perfectly correlated, i.e. $\forall j: a^j = b^j$. And when $y^2_i=1$, the measurement outcomes of $\V$ and $\Pp$ are perfectly anti-correlated, i.e. $\forall j: a^j \neq b^j$. As the verification of $\V$ only relies on this equality (or inequality) testing, we have:
    \begin{equation}
        \Prb[\V \text{ accept}_{\Pp}] = \Prb[\texttt{acc} \leftarrow \V(a^j) | \text{no Adv}: \forall j] = 1
        \label{eq:p1-completness}
    \end{equation}
    Hence the claim.
\end{proof}

\begin{protocol}
    \begin{minipage}{0.45\textwidth}
    \justifying
    \noindent The protocol runs between two parties, a \textit{Verifier} $\V$ and a \textit{Prover} $\Pp$. The protocol includes an HEPUF construction, according to Definition~\ref{hepuf}, with the set $\{ \ket{\Phi^+}, \ket{\Psi^-}$.

    \begin{enumerate}
        \item \textbf{Setup phase:}
        \begin{enumerate}
             \item $\V$ has access to the HEPUF in $\mode=0$. $\V$ constructs a database of pairs $\mathtt{DB}:= \{(x_i,y_i)\}_{i=1}^{d}$ where $y_i = f(x_i)= y_i^1||y_i^2$. $\V$ sets the $\mode=1$ and send the device to $\Pp$.
        \end{enumerate}
        \item \textbf{Authentication phase:}
        \begin{enumerate}
            \item \label{a} $\V$ chooses a pair $(x_i,y_i) \in \mathtt{DB}$ uniformly at random (per round). Then $\V$ sends $x_i$ to $\Pp$ through a classical public channel. 
            \item \label{b} $\Pp$ interacts with the HEPUF in $\mode=1$, querying it with input $x_i$. The HEPUF construction internally produces $y_i=f(x_i)=y_i^1||y_i^2$, performs the encoding $\{ (y_i^2=0)\rightarrow \ket{\Phi^+}\}$, and $\{(y_i^2=1)\rightarrow \ket{\Psi^-}\}$, and outputs verifier's part of the encoded state $\rho^i_V$.
            \item \label{c} $\Pp$ sends $\rho^i_V$ to $\V$ over the quantum channel.
            \item \label{d} $\Pp$ sets the HEPUF to $\mode=2$, resulting in measuring the local subsystem according to $y_i^1$ (If $y_i^1 =0$, measure Z basis, else if $y_i^1 =1$, measure in Hadamard basis). $\Pp$ also receives from HEPUF the corresponding outcome of each measurement $b = b^0,\dots,b^j,\dots b^m$.
            \item \label{e} $\Pp$ sends $b$ to $\V$ over the channel. 
        \end{enumerate}
        \item \textbf{Verification phase:}
        \begin{enumerate}
            \item \label{f} $\V$ measures their local qubit, in the basis of $y_i^1$ from $\mathtt{DB}$, and obtains the measurement outcome string $a = a^0,\dots,a^j,\dots a^m$. \\
            $\texttt{For j=0 to j=m}$:\\
            $\texttt{If}$ $y^2_i = 0$, $\texttt{AND}$ $a^j = b^j$, $\texttt{OR}$\\
            $\texttt{If}$ $y^2_i = 1$, $\texttt{AND}$ $a^j \oplus b^j = 1$\\
            $\texttt{Then}$ $\V$ accepts. Otherwise $\V$ rejects.
        \end{enumerate}
    \end{enumerate}
    \end{minipage}
    \caption{Online protocol with Bell states}
    \label{online_prot_2Bell}
\end{protocol}

\begin{theorem}
    Protocol \ref{online_prot_2Bell} satisfies security against QPT adversaries and unbounded quantum eavesdroppers (during authentication).
\end{theorem}
\begin{proof}
Our proof strategy is as follows. First, we perform the security analysis per round, against a quantum eavesdropper or an unbounded network adversary, without any extra information on the HEPUF. Then we allow the adversary to gather the information of HEPUF over rounds, and we show that the success probability using this extra information can be boosted by at most a negligible quantity.
For the first part, we need to show the following:
\begin{equation}
\begin{split}
    \Prb[\V \text{ accept}_{\A}] = \Prb[\texttt{acc} \leftarrow \texttt{ver}(\rho^j_{\A}, b^j_a) : (\rho^j_{\A}, b^j_a) \leftarrow \A, \forall j]
\end{split}
    \label{eq:p2-sec-1}
\end{equation}
Since the \textit{Verifier} expects a state and a classical outcome as inputs to run their verification algorithm and authenticate the \textit{Prover}, an adversary succeeds if they can produce such a tuple $(\rho^j_{\A}, b^j_a)$, per bit that successfully passes the verification with non-negligible probability. Due to the bit-independence assumption of the underlying classical PUF, we give the analysis for each bit, and the extension to $m$-bit outcome is straightforward. Also, for simplicity of the proof, we introduce a non-interactive quantum game, such that its winning strategy directly reduces to our protocol. The game consists of two parties, Alice and Bob and proceeds as follows, Alice flips two independent biased coins, $c_1$ and $c_2$, both with similar bias probabilities, $\Prb[c_1 = 0] = \Prb[c_2 = 0] = p_0 = 1/2 + \delta$ (and $p_1 = 1 - p$). Bob plays the game against Alice and needs to send simultaneously a quantum state $\rho$ and a classical bit $b$. Alice, upon receiving the state, measures it in $Z$-basis if $c_1 = 0$, and in $X$-basis if $c_1 = 1$. In either of the cases, Alice obtains the binary outcome $a$. Bob wins if,
\begin{equation}
\begin{split}
    & \texttt{if } c_2 = 0: a = b \text{ (correlated outcomes)};\\
    & \texttt{if } c_2 = 1: a \neq b \text{ (anti-correlated outcomes)}.\\
\end{split}
    \label{eq:p2-bob-win-cond}
\end{equation}
Our goal is to find the optimal winning strategy for Bob and the corresponding probability. This consists of Bob choosing bit $b$ to be $0$ or $1$ with probabilities $q_0$ and $q_1$ respectively, along with choosing a state $\rho_0$ for $b=0$, and $\rho_1$ otherwise. We find the best choice for these quantities to maximise the winning probability.\\

A general single-qubit density matrix can be written as, $\rho = \frac{1}{2}(I + r_x\sigma_x + r_y\sigma_y + r_z\sigma_z)$ where $r = (r_x, r_y, r_z)$ is the vector on the Bloch sphere. First, since Alice's bit is only dependent on $c_1$, and $c_1$ is an independent coin, we can calculate the probability of Alice receiving each of the outcomes (according to Born's rule), measured in a basis related to $c_1$,
\begin{equation}
\begin{split}
    & P_Z(0) = \Prb[a = 0 | c_1 = 0] = \frac{1 + r_z}{2}\\
    & P_Z(1) = \Prb[a = 1 | c_1 = 0] = \frac{1 - r_z}{2}\\
    & P_X(0) = \Prb[a = 0 | c_1 = 1] = \frac{1 + r_x}{2}\\
    & P_X(1) = \Prb[a = 1 | c_1 = 1] = \frac{1 - r_x}{2}\\
\end{split}
    \label{eq:p2-alice-probs}
\end{equation}
So for each of the density matrices $\rho_0$ and $\rho_1$, we have the measurement probabilities in terms of the vectors $(r_x, 0, r_z)$ and $(r'_x, 0, r'_y)$ (as the $y$ component does not contribute to the probability). The overall winning probability of Bob in this game can be written as follows,
\begin{equation}
\begin{split}
    \Prb_{win} & = \Prb[a = b|c_2 = 0] + \Prb[a \neq b|c_2 = 1]\\
    & = \Prb[a = b] \Prb[c_2 = 0] + \Prb[a \neq b] \Prb[c_2 = 1]\\
    & = p_0(\Prb[a = b, c_1 = 0] + \Prb[a = b, c_1 = 1])\\
    & + p_1(\Prb[a \neq b, c_1 = 0] + \Prb[a \neq b, c_1 = 1])
\end{split}
    \label{eq:p2-bob-win-probs}
\end{equation}
We now separate the cases for $a=b$, and for $a\neq b$. For the first one, we have two cases:
\begin{equation}
\begin{split}
    \Prb[a = b, c_1 = 0] & = \Prb[c_1 = 0]\times \sum_b q_b \Prb[a = b | c_1 = 0, \rho_b, b]\\
    & = p_0 (q_0\times P^0_Z(0) + q_1 \times P^1_Z(1))
\end{split}
    \label{eq:p2-bob-win-aeqb-0}
\end{equation}
and similarly for $c_1=1$ we have
\begin{equation}
    \Prb[a = b, c_1 = 1] = p_1 (q_0\times P^0_X(0) + q_1 \times P^1_X(1))
    \label{eq:p2-bob-win-aeqb-1}
\end{equation}
and similarly for both cases of $a\neq b$:
\begin{equation}
\begin{split}
    & \Prb[a \neq b, c_1 = 0] = p_0 (q_0\times P^0_Z(1) + q_1 \times P^1_Z(0))\\
    & \Prb[a \neq b, c_1 = 1] = p_1 (q_0\times P^0_X(1) + q_1 \times P^1_X(0))
\end{split}
    \label{eq:p2-bob-win-aneqb}
\end{equation}
Substituting all these in Eq.~\ref{eq:p2-bob-win-probs} and putting rearranging we have:
\begin{equation}
\begin{split}
        \Prb_{win} = \frac{1}{2}& \times [(\frac{1}{2} + \delta)^2 \{1 - r'_z + q_0 (r_z + r'_z)\}\\
        & + (\frac{1}{2} + \delta)(\frac{1}{2} - \delta)\{2 + (r'_z - r'_x) \\
        & + q_0(r'_x - r'_z + r_x - r_z)\}\\
        & + (\frac{1}{2} - \delta)^2 \{1 + r'_x - q_0 (r_x + r'_x)\}]
\end{split}
    \label{eq:p2-bob-win-final}
\end{equation}
We need to maximise this function over all $5$ parameters $\{r_x,r_z,r'_x,r'_z, q_0\}$, given the constraints of $|r_x| \leq 1$, $|r_z|\leq 1$, $|r'_x|\leq 1$, $|r'_z| \leq 1$, and $r_x^2 + r_z^2 \leq 1$, ${r'_x}^2 + {r'_z}^2 \leq 1$, and finally, $0 \leq q_0 \leq 1$. We first maximise this equation in terms of $q_0$. It can be seen that the optimal choice is to pick $q_0$ based on $\max\{r_z,r_x\}$, i.e $q_0 = 1$ if $r_z > r_x$ and $q_0 = 0$ if $r_z < r_x$. Here we omit the proof of this constraint optimisation (which we have done in the polar coordinates) and we only give the result. We conclude that the above function does not have a critical point in the desired boundaries, and given all the constraints, since the optimal solution for $q_0$ (which is either 0 or 1) and the solution for other parameters cannot be simultaneously satisfied. In fact, by doing so, we have managed to prove that the probabilistic strategy of picking the appropriate density matrix with probabilities $q_0$ and $q_1$ does not admit the optimal strategy. Thus, we can set the probability $q_0$ to one of the extreme points (Without loss of generality, we pick $q_0 = 1$), and select an optimal density matrix, according to that deterministically. \\
Assuming the deterministic strategy and picking the $q_0 = 1$, we can rewrite the winning probability as follows,
\begin{equation}
\begin{split}
        \Prb_{win} = \frac{1}{2}& \times [(\frac{1}{2} + \delta)^2 (1 + r_z)\\
        & + (\frac{1}{2} + \delta)(\frac{1}{2} - \delta)(2 + r_x - r_z)\\
        & + (\frac{1}{2} - \delta)^2 (1 - r_x)]
\end{split}
    \label{eq:p2-bob-win-final}
\end{equation}
Now we optimise this function again over the choice of the density matrix, i.e. over $r_x$ and $r_z$, given all the Bloch sphere constraints above. The following parameters maximise the function:
\begin{equation}
\begin{split}
        r^{opt}_x = \frac{1 - 2\delta}{\sqrt{2 + 8\delta^2}}, \quad r^{opt}_z = \frac{1 + 2\delta}{\sqrt{2 + 8\delta^2}}
\end{split}
    \label{eq:p2-optimal-vector}
\end{equation}
so the optimal density matrix of Bob is
\begin{equation}
    \rho^{opt} = \frac{1}{2}(I + r^{opt}_x \sigma_x + r^{opt}_z \sigma_z)
    \label{eq:p2-opt-density}
\end{equation}
Which gives the final optimal probability:
\begin{equation}
    \Prb_{win} = \frac{1}{2} + \delta \sqrt{\frac{1 + 4\delta^2}{2}}
    \label{eq:p2-Bob-win-opt}
\end{equation}
Giving the extreme cases of $\delta = 0 \rightarrow \Prb_{win} = 1/2$, and $\delta = 1/2\rightarrow \Prb_{win} = 1$ as expected.\\
Without any extra information, we claim that the success probability of any unbounded channel adversary in passing the verification of Protocol~\ref{online_prot_2Bell} is the same as the success probability of Bob in winning this game. To complete the proof, we need to show that there is a bound on the amount of information the adversary can get through every pair of state and classical bit, meaning that an overall QPT adversary can only boost the probability of guessing $y^1$ and $y^2$ by a negligible factor. 

First, we use the local-indistinguishability property of the state to show that, through the state alone, the adversary does not learn any extra information about the outcome bit of the PUF.

Let $p_{guess}$ be the guessing probability of $y^2$  based on a single copy of the reduced density matrix $\rho(y^2)$ that the adversary receives.  
    \begin{equation}
        \rho(y^2) = \tr_{\Pp}[\ket{\psi(y^2))} \bra{\psi(y^2))}]
    \end{equation}
with $\ket{\psi(y^2))} \in \{ \ket{\Phi^+},\ket{\Psi^-}\}$ according to the construction. As the ensemble of the state depends only on $y^2$, the choice of measurement basis depends only on $y^1$, prior to any announcement by $\Pp$, and only by disrupting the quantum channel, density matrices corresponding to $y^{j=1}=0$ and $y^{j=1}=1$, from the point of view of the adversary, are:
    \begin{eqnarray}
    \rho_0 &&=\rho(y^2)|_{y^1=0} \nonumber \\
    &&= (\frac{1}{2} + \delta)(\tr_{\Pp} [\ket{\Phi^+}\bra{\Phi^+}])+(\frac{1}{2} - \delta)(\tr_{\Pp} [\ket{\Psi^-}\bra{\Psi^-}]) \nonumber\\
    &&= \frac{\openone}{2}. 
    \end{eqnarray}
    and, 
    \begin{eqnarray}
    \rho_1 &&=\rho(y^2)|_{y^1=1}\nonumber\\
    && = (\frac{1}{2} + \delta)(\tr_{\Pp} [\ket{\Phi^+}\bra{\Phi^+}])+(\frac{1}{2} - \delta)(\tr_{\Pp} [\ket{\Psi^-}\bra{\Psi^-}]) \nonumber\\
    && =\frac{\openone}{2}.
    \end{eqnarray}
Note that, in both cases, the state is a maximally mixed state, leading to the guessing probability of $y^2$ for the adversary being,
\begin{equation}
    p^{on}_{guess}(y^2) \coloneq \Prb [\mathcal{A}^{j}_{guess}(x, \rho(y^2))=y^2] = \frac{1}{2}.
\end{equation}
As all the local states are identical and maximally mixed, the probability of identifying them is no better than a random guess. \cite{helstrom69, ivanovic87, peres98}. We conclude that the adversary does not learn any extra information about the state to boost their success probability.
    
Next, we consider that the adversary obtains the local state along with the classical information to capture all the possible strategies of the adversary, given all the existing information on the channel. For that, we need to consider the following types of attacks where $\A$ can query $\Pp$ with their choice of $x_i$, pretending to be the \textit{Verifier} and then wait for the \textit{Prover} to perform the measurement and announce the classical outcome. As the choice of measurement basis solely depends on the value of $y^1$, and we have shown that the local states (even given the bias) are maximally mixed, we only consider the probability of $\A$ guessing $y^1$ through this attack, as follows:
    \begin{equation}
        p^{on}_{guess}(y^1) \coloneq \Prb [\mathcal{A}^{j}_{guess}(x, \rho(y^2))=y^1].
    \end{equation}
    When $y^1=0$, the measurement performed by the HEPUF is in the computational basis. In this case, we denote the projectors by $\Pi_0=\ket{0}\bra{0}$ and $\Pi_1=\ket{1}\bra{1}$. The probability of obtaining the outcome $``0"$ from $\Pi_0$ is given by, $p(\Pi_0)=\tr[(\openone \otimes\Pi_0) ((1/2 + \delta) \ket{\Phi^+}\bra{\Phi^+}+(1/2 - \delta) \ket{\Psi^-}\bra{\Psi^-})] = 1/2$. Similarly, the probability of obtaining the outcome $``1"$ is given by $p(\Pi_1)=1/2$. When the projector $\Pi_0$ clicks, the announced outcome is $``0"$, and the local state on the adversary's side collapses to,
    \begin{eqnarray}
        \rho'_0 &&= \frac{1}{p(\Pi_0)} \tr_{\Pp} [(\openone \otimes \Pi_0) ((1/2 + \delta)\ket{\Phi^+}\bra{\Phi^+} \nonumber\\
        && \quad \quad \quad \quad \quad \quad +(1/2 - \delta)\ket{\Psi^-}\bra{\Psi^-})(\openone \otimes \Pi_0)] \nonumber\\
        && = 2 \tr_{\Pp} [\frac{(1/2 + \delta)}{2} \ket{00}\bra{00}+\frac{(1/2 - \delta)}{2}\ket{10}\bra{10}]\nonumber\\
        && = (1/2 + \delta) \ket{0}\bra{0}+(1/2 - \delta) \ket{1}\bra{1}.
    \end{eqnarray}
    And when the outcome is $``1"$ corresponding to the projector $\Pi_1$, the post-measurement state at the adversary's side becomes,
    \begin{eqnarray}
        \rho'_1 &&= \frac{1}{p(\Pi_1)} \tr_{\Pp} [(\openone \otimes \Pi_1) ((1/2 + \delta)\ket{\Phi^+}\bra{\Phi^+} \nonumber\\
        && \quad \quad \quad \quad \quad \quad +(1/2 - \delta)\ket{\Psi^-}\bra{\Psi^-})(\openone \otimes \Pi_1)] \nonumber\\
        && = (1/2 + \delta) \ket{1}\bra{1} + (1/2 - \delta) \ket{0}\bra{0}.
    \end{eqnarray}

Similarly, when $y^1=1$, the measurement is performed in the Hadamard basis with the projectors $\Pi'_0=\ket{+}\bra{+}$ and $\Pi'_1=\ket{-}\bra{-}$. In this scenario, the probabilities for obtaining outcomes $b=0$ and $b=1$ corresponding to the projectors $\Pi'_0$ and $\Pi'_1$ respectively are $p(\Pi'_0)=p(\Pi'_1)=1/2$. Due to this measurement, the pose-measurement state on the adversary's side collapses to either of the following with equal probability,
\begin{equation}
    \begin{split}
        & \rho'_2 = (1/2 + \delta) \ket{+}\bra{+}+ (1/2 - \delta) \ket{-}\bra{-} \\
        & \rho'_3 = (1/2 + \delta) \ket{-}\bra{-}+ (1/2 - \delta) \ket{+}\bra{+}
    \end{split}
\end{equation}

\begin{figure}[htp]
\centering
\fbox{
\subfigure[Success probability of the adversary per round with respect to the bias in the CPUF.]{\includegraphics[scale=0.25]{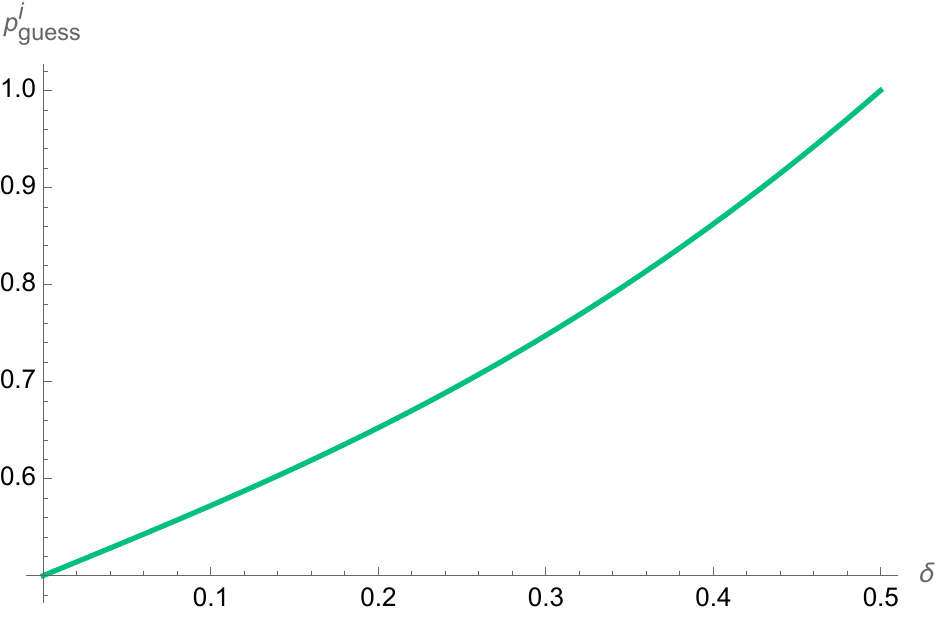}}
\quad
\subfigure[Success probability of the adversary after $m$ rounds for given bias in the CPUF.]{\includegraphics[scale=0.25]{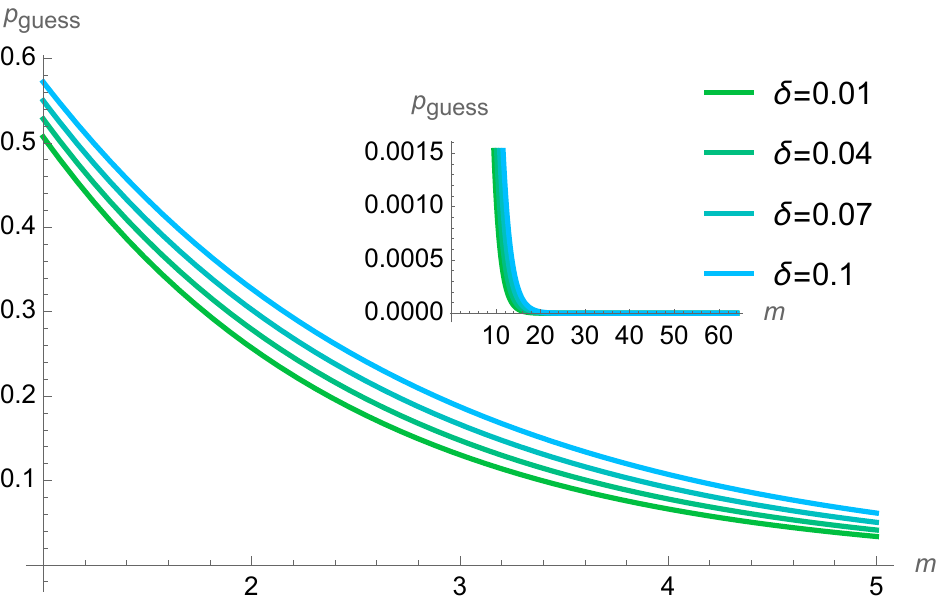}}}
\caption{\footnotesize{The success probability of the adversary to extract the database per round and after $m$ rounds. All the axes are unitless.}}
\label{prob_on}
\end{figure}

\begin{table*}
    \setlength{\tabcolsep}{5pt}
    \renewcommand{\arraystretch}{1.5}
    \newcolumntype{C}[1]{>{\centering\arraybackslash}p{#1}} 
    \begin{tabular}{|C{3cm}|C{2cm}|C{2cm}|C{1.5cm}|C{1.7cm}|C{5cm}|}
         \hline
         The protocol & Underlying primitive & Quantum communication & Memory & Trusted resource & Security parameter (after $m$ rounds) \\
         \hline \hline
         Ref.~\cite{arapinis21} & QPUF & $2$-way & Quantum & -- & $(\frac{1}{2})^m$ \\
         \hline
         Ref.~\cite{chakraborty23} & HPUF & $2$-way & Classical & -- & $p^{cl}_{f} ((\frac{1}{2}+\delta)(1+\sqrt{2} (\frac{1}{2}+\delta)))^{2m\times poly(m)}$ Th.45 of~\cite{doosti2022unclonability} \\
         \hline
         This work (offline protocol) & CPUF & -- & Classical & Bell states & $(\frac{1}{2})^m$ \\
         \hline
         This work (online protocol) & HEPUF & $1$-way ($\Pp~\rightarrow~\V$) & Classical & -- & $\left(\frac{1}{2} + \delta \sqrt{\frac{1 + 4\delta^2}{2}}\right)^m$ \\
         \hline
    \end{tabular}
    \caption{Comparison between the existing hardware-based authentication protocols. In the second row, $p^{cl}_{f}$ denotes the forging probability of a classical PUF given $q = poly(m)$ queries.}
    \label{comparison}
\end{table*}

When the announcement by the \textit{Prover} is $b=0$, the local state of the adversary is either of $\rho'_0$ or $\rho'_2$, and when the announcement is $b=1$, the local state is either of $\rho'_1$ or $\rho'_3$. Hence, for a given announcement, say $b=0$, the correct guessing probability of $y^1$ is essentially the success probability of distinguishing between the states $\rho'_0$ and $\rho'_2$. This is given by the Helstrom bound \cite{helstrom69} as follows,
\begin{equation}
    \begin{split}
        p_{guess}^{on(b=0)}(y^1) = & \frac{1}{2}+\frac{1}{2}||p \rho'_0-(1-p)\rho'_2||_1\\
        & \frac{1}{2} + \frac{\delta}{2}(|1 + \sqrt{\frac{1 + 4\delta^2}{2}}| + |1 - \sqrt{\frac{1 + 4\delta^2}{2}}|)\\
        & \frac{1}{2} + \delta
    \end{split}
\end{equation}
with $0 \leq \delta \leq 1/2$. Similarly, for $b=1$ the adversary needs to distinguish between the states $\rho'_1$ and $\rho'_3$. A simple calculation shows $p_{guess}^{on(b=1)}(y^1)=p_{guess}^{on(b=0)}(y^1)= \frac{1}{2} + \delta$. This recovers the initial guessing probability of $y^1$ given the bias $\delta$, and $\A$ cannot get any extra information through this type of attack to boost their success probability from the one in Eq.~\ref{eq:p2-Bob-win-opt}. We can now write the final optimal success probability of $\A$, for an $m$-bit outcome PUF, as
\begin{equation}
    \Prb[\V \text{ accept}_{\A}] = \left(\frac{1}{2} + \delta \sqrt{\frac{1 + 4\delta^2}{2}}\right)^m \approx negl(m).
\end{equation}
This concludes the proof.
\end{proof}
The exact variation of the optimal success probability for the adversary per round can by found in Fig.~\ref{prob_on}(a) and after $m$ rounds in Fig.~\ref{prob_on}(b).

\section{Discussion and Future Directions}
One of the key challenges that makes any cryptographic protocol vulnerable to external attacks is the authentication of the parties involved in the process~\cite{Ling25}. We design an entanglement-based hybrid authentication protocols for secure quantum communication. This is the first entangled version of a hardware-based hybrid communication protocol. Our construction includes both offline and online authentication schemes, using two-qubit maximally entangled states.

In the offline protocol, we propose a fast authentication scheme for honest parties, which relies solely on classical communication once the protocol is initiated. This scheme assumes that several entangled states have been pre-shared between the involved parties. These states are distributed by a trusted source well before the authentication begins. We show that the offline authentication protocol is complete and secure against unbounded quantum network adversaries. In our second protocol—the online version - we remove the trusted source and instead rely on a fundamental property of entangled states. By considering a locally indistinguishable (LI) set of states, we construct a new authentication protocol based on what we call a hybrid entangled PUF. In this setting, outputs of a CPUF are encoded into bipartite maximally entangled states chosen from a predefined LI set. These states are indistinguishable under LOCC operations, and security is ensured by the local indistinguishability property, as the parties are limited to LOCC on their respective subsystems.

We show that, using a weak CPUF and assuming it is not fully broken, the online protocol is complete and exponentially secure, even against the strongest adversaries, who are unbounded over the network and QPT overall. We analyse the security scaling of the protocol as a function of the CPUF's bias parameter and demonstrate that even in the presence of a non-negligible bias, the protocol achieves exponential security in a single round.

While our authentication protocol can also be extended to key agreement schemes, we leave the analysis and formal security proof of such an extension for future work. Another desirable feature of our protocol, inspired by~\cite{chakraborty23}, is challenge re-usability. Although we do not provide a formal proof of this property in the current work, we believe it is likely achievable due to the properties of the LI states.

Finally, we believe experimental implementations of these protocols will serve as compelling proof-of-concepts demonstrating the power of hybrid communication schemes. A detailed comparison between existing PUF-based authentication protocols and our results is provided in Table~\ref{comparison}.

\section*{Acknowledgements}
The authors thank Ramin Jafarzadegan and Alexandru Cojocaru for valuable discussions and comments during different stages of this work. All the authors acknowledge the support of the Quantum Advantage Pathfinder (QAP), with grant reference EP/X026167/1, and the UK Engineering and Physical Sciences Research Council. MD and EK also acknowledge the Integrated Quantum Networks Hub, grant reference EP/Z533208/1.

\bibliography{references.bib}

\end{document}